\documentclass[12pt]{amsart}
\usepackage{amssymb, amsmath, amsthm}
\DeclareMathAlphabet{\txcal}{U}{tx-cal}{m}{n}
\usepackage[foot]{amsaddr}
\allowdisplaybreaks
\usepackage[table,xcdraw]{xcolor}
\usepackage{multirow}
\usepackage[left=2.5cm, right=2.5cm, top=3cm, bottom=3cm]{geometry}
\PassOptionsToPackage{hyphens}{url}
\usepackage[
   colorlinks = true,
   urlcolor = blue,
   linkcolor = blue,
   citecolor = blue]{hyperref}
\usepackage[utf8]{inputenc}
\newcommand{\etalchar}[1]{$^{#1}$}

\newcommand{\SNP}{\text{Poly}}
\newcommand{\Poly}{\text{Poly}}
\newcommand{\Prb}{\mathbb P}

\newcommand{\sS}{\mathbb S}

\newtheorem{lemma}{Lemma}
\usepackage{graphicx}

\begin{document}
\title[SARS-Cov-2 variant composition in pooled samples]{A mixture
  model for determining SARS-Cov-2 variant composition in pooled
  samples} 
\author[R. Valieris, R. Drummond, A. Defelicibus]{Renan Valieris$^1$,
  Rodrigo Drummond$^1$,  Alexandre Defelicibus$^1$} 
\address{$^1$Laboratory of Computational Biology and Bioinformatics,
  CIPE/A.C. Camargo Cancer Center, São Paulo 01508-010, Brazil}
\email{valieris@accamargo.org.br}
\email{rdrummond@accamargo.org.br}
\email{alexandre.defelicibus@accamargo.org.br}
\author[E. Dias-Neto]{Emannuel Dias-Neto$^2$}
\address{$^2$Laboratory of Medical Genomics,  CIPE/A.C. Camargo Cancer
  Center, S\~ao Paulo 01508-010, Brazil}
\email{emmanuel@cipe.accamargo.org.br}
\author[R.A. Rosales]{Rafael A. Rosales$^3$}
\address{$^3$Departameto de Computa\c{c}\~ao e
  Matem\'atica, Universidade de S\~ao Paulo,
  Avenida Bandeirantes 3900, Ribeir\~ao Preto, S\~ao Paulo,
  14040-901, Brazil}
\email{rrosales@usp.br}
\author[I. T. Silva]{Israel Tojal da Silva$^1$}
\email{itojal@accamargo.org.br}

\date{\today}

\begin{abstract}
Despite of the fast development of highly effective vaccines to
control the current COVID$-$19 pandemic, the unequal distribution and
availability of these vaccines worldwide and the number of people
infected in the world lead to the continuous emergence of SARS-CoV-2
(Severe Acute Respiratory Syndrome coronavirus 2) variants of
concern. It is likely that real-time genomic surveillance will be
continuously needed as an unceasing monitoring tool, necessary to
follow the spillover of the disease spread and the evolution of the
virus. In this context, new genomic variants of SARS-CoV-2 that may
emerge as a response to selective pressure, including variants
refractory to current vaccines, makes genomic surveillance programs
tools of utmost importance. Here propose a statistical model for the
estimation of the relative frequencies of SARS-CoV-2 variants in
pooled samples. This model is built by considering a previously
defined selection of genomic polymorphisms that characterize
SARS-CoV-2 variants. The methods described here support both raw
sequencing reads for polymorphisms-based markers calling and
predefined markers in the VCF format. Results obtained by using
simulated data show that our method is quite effective in recovering
the correct variant proportions. Further, results obtained by
considering longitudinal data from wastewater samples of two locations
in Switzerland agree well with those describing the epidemiological
evolution of COVID-19 variants in clinical samples of these
locations. Our results show that the described method can be a
valuable tool for tracking the proportions of SARS-CoV-2 variants. 
\end{abstract}

\maketitle

\section{Introduction}
The astonishing speed seen for the global spread of COVID-19 has
prompted a large global effort to control this outbreak. The first
complete SARS-CoV-2 genome was published on January 05, 2020 by
\cite{pmid32015508}. Thenceforth, SARS-CoV-2 sequences recovered from
patients from most countries have been made available to the
scientific community \cite{pmid33057582}, allowing a better
understanding of the geographical and temporal spreading of
SARS-CoV-2, including the indication of non-synonymous genomic
variants that may explain the increased replication rate and immune
escaping of some variants \cite{pmid34395364} Today, the Global
Initiative on Sharing Avian Influenza Data ({\bf GISAID}) database is
arguably the primary archive for SARS-CoV-2 genome sequences
\cite{pmid33893460}. 

Although vaccination is the most effective means of preventing
COVID-19 illnesses and related deaths \cite{pmid34324480}, additional
efforts employing genomic surveillance have proven to be a useful tool
for guiding upcoming measures to control virus transmission
\cite{pmid34337584, pmid29129921}. The sequencing of Viral RNA
genomes directly recovered from wastewater has recently gained
attention for providing an opportunity to assess circulating viral
lineages \cite{pmid34323811, pmid34058182, pmid33501452,
  pmid32942178,Jahn2021}. These studies take advantage of
shotgun-based sequencing protocols \cite{pmid28337072} followed by
the most typical computational workflows to unveil the genomic
diversity present in a given sample, consisting of (1) sequencing
quality profiling; (2) removal of host/rRNA data; (3) assembly of
reads; and (4) attribution of taxonomy \cite{pmid28337072,
  pmid33416890}. It is also be clear that this approach can be applied
in other sources of urban metagenomic surveilance \cite{Dankoetal}. 

The main contribution of this article consists in the development of a
statistical model to infer the relative proportions and frequencies of
the genomic variants of SARS-CoV-2 present in varying amounts in a
given sample. This task is far from trivial as the sequencing reads
deriving from a sample consist of relatively short sequences ($\approx
100-200$ bases long) that can be mapped to multiple variants of the
virus. The sequencing reads derived from a given locus of the viral
genome may be different across the individuals of the same variant due
to intra-clade variation. Also, the small proportion of sequenced
reads likely to align to the SARS-CoV-2 in the midst of a complex
mixture of other RNAs (human, viral, bacterial and others) may lead to
reduced vertical coverage of the viral genome, therefore decreasing
the likelihood of an effective variant monitoring. 

Here, we propose a viral composition deconvolution approach based on
the relative frequencies of genomic polymorphic markers found in
SARS-CoV-2 variants. These markers, either Single Nucleotide
Polymorphisms (SNPs) or INsertion or DELetion of bases (INDELs), are
selected from public SARS-CoV-2 data \cite{pmid28382917,
  pmid33972780} and their presence/absence in known SARS-CoV-2
variants is used to fit a mixture model to viruses, derived from
different subjects, found in a complex mixture. Following this, our
method calculates a maximum likelihood estimate of the relative
contributions of SARS-CoV-2 variants to the pool. The performance of
the test was evaluated in simulated and real data. Our analysis using
122 sequencing data-points from wastewater treatment plants collected
in Switzerland, show close correlation with epidemiological trends of
COVID-19 in that region, which demonstrates the utility of this
approach to guide public policies 

\section{Variant composition model}

Our model is built over a previously defined selection of genomic
polymorphisms, which characterize SARS‑CoV‑2 variants, and a matrix
$P$ of `variant signatures'. Formally, let $P = (P_{ij})$ be a $s
\times v$ matrix such that $P_{ij}$ corresponds to the probability of
finding an alternate sequence at polymorphism $i$ from variant
$j$. Details about the selection of polymorphisms of interest and the
construction of $P$ are provided in Section~\ref{sec:markers}. Given
the matrix $P$ and a data sample containing the counts of DNA fragment
readings aligned to the respective polymorphic loci at SARS-CoV-2
genome, we aim at estimating the vector $w = (w_1, w_2, \ldots, w_v)$
of the relative compositions of SARS-CoV-2 variants in the sample,
such that
\begin{equation}\label{eqn:freqs_sum}
  0 \leq w_j \leq 1, \qquad \sum_{j=1}^v w_j = 1.
\end{equation}

Let $\{\Poly_i\}$, $i=1, 2, \ldots, v$, be the set of polymorphisms of
interest. The data provided by a sample consists of counts of reads
supporting the reference sequence $c_i^r$ and those supporting an
alternate sequence, $c_i^a$ at each $\Poly_i$. Two crucial remarks
about the data are that the coverage, that is $c_i^a+c_i^r$, may vary
among polymorphisms, and that the data does not shows which variants
account for the actual observed counts $c_i^r$ and $c_i^a$. The latter
occurs because:
\begin{itemize}
\item the reads that constitute the sample are
  relatively small segments of the viral genomes,
\item variants may share few $\Poly_i$ events,
\item there is random variability among the polymorphisms across the
  individuals of the same variant.
\end{itemize}

A relation between counts and variants is made by introducing the
latent variables $Z_{ij}^r$ and $Z_{ij}^a$, representing respectively
the counts of reads bearing the reference sequence for a $\SNP_i$
event originating from variant $j$ and those bearing the reference
base at the position $i$ from the variant $j$. In this case
\begin{equation}
  \label{eqn:latent&counts}
  \sum_{j=1}^v Z_{ij}^r = c_i^r
  \qquad\text{and}\qquad
  \sum_{j=1}^v Z_{ij}^a = c_i^a.
\end{equation}
Let $C_i^r$ and $C_i^a$ be random variables which for a given sample
take on the values $c_i^r$ and $c_i^a$. Let $\txcal{Z}$ be the set of
possible values of $Z_{ij}^a$ and $Z_{ij}^r$ satisfying the
constraints imposed by \eqref{eqn:latent&counts}.  For each $i = 1, 2,
\ldots, s$, let $C_i = (C_i^r, C_i^a)$ and $Z_i = (Z_i^r, Z_i^a)$,
where $Z^n_i = (Z_{i1}^n, \ldots, Z_{iv}^n)$, $n \in \{a, r\}$. Denote
by $t = (t_1, t_2, \ldots, t_s)$ the coverage vector, namely a vector
such that $t_i$ corresponds to the total number of reads at the locus
of $\SNP_i$ observed in given sample, $t_i = c_i^a + c_i^r$. The
values taken by the latent variables are denoted by using lower case
symbols accordingly. Assuming independence between the events
$\SNP_i$, $i = 1, 2, \ldots, s$, the likelihood function for $w$ at a
given sample and a given $P$ matrix is determined by
\[
  L(w) = \Prb\big(C = c \mid w, P, t\big)
  =
  \prod_{i=1}^s\prod_{n\in \{r, a\}} \Prb\big(C_i^n = c_i^n \mid w,
  P, t_i\big)
\]
Thus, by considering the latent variables, the likelihood can be
written as
\begin{align*}
  L(w)
  &=
  \prod_{i=1}^s\prod_{n \in \{r, a\}}\sum_{z_i \in \txcal{Z}}
    \Prb\big(C_i^n = c_i^n \mid  Z_i^n=z_i^n, w, P, t_i\big)
    \Prb\big(Z_i^n = z_i^n \mid w, P, t_i\big)
    \\
  &=
    \prod_{i=1}^s\prod_{n \in \{r, a\}} \sum_{z_i \in \txcal{Z}}
    \mathbf{1}_{\big\{\sum_j Z_{ij}^n =
    c_i^n\big\}}
    \Prb\big(Z_i^n = z_i^n \mid w, P, t_i\big),
\end{align*}
where the summation over $z_i \in \txcal{Z}$ considers all possible
values for the latent variables subject to the constraints
\eqref{eqn:latent&counts}. 

Since the sequencing process picks DNA fragments at random from the
studied pool, let us assume that the distribution of the latent
variables is multinomial with parameters given by the relative
proportions of RNAs from each variant, supporting or not a
mutation. If the proportion of the $j$-th variant is $w_j$ and the
probability that this variant presents a mutation at position $i$ is
$P_{ij}$, then the fraction of total RNA originated from variant $j$
supporting an altered base at $\SNP_i$ equals $P_{ij}w_j$. Likewise,
the fraction of total RNA originated from variant $j$ supporting the
reference base at $\SNP_i$ is $(1-P_{ij})w_j$. So, conditionally on
$w$, $P$ and $t_i$, the law of the latent variables, $\Prb\big(Z_i =
z_i \mid w, P, t_i\big)$, equals
\[
  \bigg({t_i \atop z_{i1}^a, \ldots, z_{iv}^a, z_{i1}^r, \ldots,
    z_{is}^r}\bigg) 
  \prod_{j=1}^v \big(P_{ij}w_j\big)^{z_{ij}^a}
  \big((1-P_{ij})w_j\big)^{z_{ij}^r}.
\]

As shown by the following lemma, the log-likelihood for the model of
variant proportions presented so far admits a closed form. The proof
of this lemma is presented in the Appendix.

\begin{lemma}\label{lem:loglik} For a given sample with counts $c_i$,
$ 1 \leq i \leq s$, and a given signature matrix $P$, the log
likelihood function for the variant proportions up to a constant
equals
\begin{equation}\label{eqn:loglik}
  \ell(w) \propto \sum_{i=1}^s c_i^a \log\bigg(\sum_{j=1}^v
  P_{ij}w_j\bigg) + \sum_{i=1}^s c_i^r\log\bigg(1 - \sum_{j=1}^v
  P_{ij}w_j\bigg).
\end{equation}
\end{lemma}

Estimates for $w$, the proportion of each variant in a sample, are
obtained by maximization of $\ell(w)$ in \eqref{eqn:loglik}. These are
hereafter denoted by $\widehat{w}$ and eventually, to emphasize their
dependence upon the sample $c$, also by $\widehat{w}(c)$. The
maximization of $\ell(w)$ is made as described in
Section~\ref{sec:maximization}. Standard error estimates for $w$ are
obtained via bootstrapping, see Section~\ref{sec:bootstrap}. 
 
\section{Methods}

\subsection{Variant characterization by polymorphism-based
  markers}\label{sec:markers} 

This study takes advantage of publicly available data from {\bf
  GISAID}~\cite{pmid28382917}. Pre-defined SARS‑CoV‑2 lineages were
assigned to variant groups (denoted hereafter as \emph{VG}, see
supplementary Table \ref{tab:table-vg}), according to variants
currently defined by the World Health Organization ({\bf WHO},
\url{https://www.who.int/en/activities/tracking-SARS-CoV-2-variants/}). Next,
a list of manually curated genome designations (v1.2.60,
\url{https://github.com/cov-lineages/pango-designation}) was obtained
from the Pango Lineage Designation Committee \cite{pmid32669681} and
each genome was assigned to the corresponding {\bf VG}. These genomes
were mapped to the SARS-CoV-2 reference genome \cite{pmid32015508}
using minimap2 v2.22 with map-pb mode \cite{pmid29750242}. {\bf GATK}
Mutect2 v4.2.2.0 was then used with default settings
\cite{pmid21478889} to call all polymorphisms (SNP and INDELs) in the
aligned sequences for each {\bf VG}. Finally, polymorphism-based
markers, denoted hereafter as {\bf PBM}, were extracted from
well-characterized variants (Figure \ref{fig:markers}). Alternatively,
we also created a equivalent marker matrix from the phylogenetic tree
compiled by \cite{pmid33972780}, available at
\url{https://hgdownload.soe.ucsc.edu/goldenPath/wuhCor1/UShER_SARS-CoV-2},
for further validation of the former matrix signature. Only
polymorphic sites with allele frequencies greater than 80\% in at
least one {\bf VG} were considered as valid markers. Although our
pipeline can provide a useful wrapper for marker calling, it is
important to note that it offers flexibility for the user to load its
own selection of markers in the variant call format ({\bf VCF}) file. 

\subsection{Wastewater dataset, dilution experiments and polymorphism
  calling}\label{sec:wastewater} 

A real data set constituted by longitudinal samples from two
wastewater treatment plants in Switzerland \cite{Jahn2021}, including
Zürich (64 samples, Jul 2020-Feb 2021)  and Lausanne (49 samples, Sep
2020-Feb 2021) were downloaded from Sequence Read Archive
(SRA,\url{https://www.ncbi.nlm.nih.gov/sra})  under project accession
number PRJEB44932. In addition, replicate dilution series experiments
containing RNA samples of  cultivated Wild type \cite{pmid32015508}
and $Alpha$/B.1.1.7 solutions mixed (ratios of 10:1, 50:1 and 100:1)
were also obtained from the same  study. The raw sequences were
previously aligned to the SARS-CoV-2 reference genome
\cite{pmid32015508} as described in \cite{Jahn2021}  and then loaded
into GATK-Mutect2 v4.2.2.0 with the optional argument -alleles
\cite{pmid21478889} to report the coverage of all PBM sites.

\subsection{Simulation study}

In order to test our method, we simulated samples with randomly
generated variant compositions. Aiming at reproducing real data, the
wastewater dataset described in Section~\ref{sec:wastewater} was used
as a template to generate the coverage distribution of simulated
samples.

For each simulated sample, a real wastewater sample was randomly
selected and its total coverage was reproduced. The distribution of
reads covering each polymorphism was generated by considering a
multinomial distribution, in which the probabilities of a read
covering each locus were proportional to the total number of reads
observed at the respective loci in the selected sample. This gives a
mock coverage $t = (t_1, t_2, \ldots, t_s)$, which stores the
simulated number of reads covering the locus of each $\SNP_i$.

Relative variant frequencies, $w^{\text{sim}}_j$, were also generated
randomly and, for each $\SNP_i$, $t_{i}$ reads were distributed among
variants according to those frequencies using a multinomial
distribution. This procedure generates $R_{ij}$, the number of reads
aligned to $\SNP_i$ originating from variant $j$. Finally, for each
$R_{ij}$, the number of reads supporting the alternate sequence at
$\SNP_i$ was generated by a Binomial$(n,p)$ distribution with $n =
R_{ij}$ and  $p = P_{ij}$. Simulated data were then obtained by
summing up the number of simulated reads supporting the reference or
the alternate sequence for each polymorphism. These were inputed to
the model and estimated compositions were compared to those used in
each simulation. The accuracy of the results was measured by the mean
absolute error
\begin{equation}\label{eqn:mae}
\frac{1}{v} \sum_{j=1}^v | \widehat{w}_j - w_j|
\end{equation}

\subsection{Likelihood maximization}\label{sec:maximization}

Let $\sS$ be the convex set defined by the unitary $(v-1)$-simplex,
that is $\sS = \{w \in \mathbb R^v : w \text{ satisfies }
\eqref{eqn:freqs_sum}\}$, with $v$ as the number of SARS-CoV-2
variants in a pool.  The following lemma ensures that the maximization
of the log-likelihood function defined by Lemma~\ref{lem:loglik} is a
well posed problem.
\begin{lemma}\label{lem:convexity} The function $\ell : \sS \to
  \mathbb R$ defined in Lemma \ref{lem:loglik} is concave.
\end{lemma}

The proof of Lemma~\ref{lem:convexity} is presented in the
Appendix. The maximization of the log-likelihood is implemented with
\emph{CVXPY} v1.1.15 \cite{diamond2016}, a Python-embedded modeling
language for convex optimization problems. \emph{CVXPY} uses
disciplined convex programming, a system for constructing mathematical
expressions with known curvature.

\subsection{Standard errors of estimates}\label{sec:bootstrap}

Estimates for the standard error of $\widehat{w}$ are obtained by
bootstrapping. For a given sample $c = (c_i^a, c_i^r)$, $i = 1, 2,
\ldots, s$, let $c^{*b}$, $b=1, 2, \ldots, B$, be a set bootstrap
replications, see \cite{ET93}, Chapter 2. A bootstrap estimate of the
standard error of $\widehat{w}$ is then given by
\[
  \widehat{\text{se}}_{\text{boot}}
  =
  \bigg\{\sum_{b=1}^B \big(\widehat{w}(c^{*b}) -
  \widehat{w}_{\text{boot}}\big)^2/(B-1)\bigg\}^{\frac12}
\]
where $\widehat{w}_{\text{boot}} = \sum_{b=1}^B
\widehat{w}(c^{*b})/B$. Results described here were obtained by using
$B=100$ bootstrap samples. 

\subsection{Code and markers availability}\label{sec:code-availability}
The markers matrix and the computational pipeline, including the
construction of variant markers, polymorphism calling, the likelihood
maximization and standard error estimates described in
Sections~\ref{sec:markers}, \ref{sec:wastewater},
\ref{sec:maximization} and \ref{sec:bootstrap}, are available at
\url{https://github.com/rvalieris/LCS}. 

\section{Results}

\subsection{Generating polymorphism-based markers of SARS-CoV-2 variants}

We have built a list of polymorphism based markers to distinguish
known SARS-CoV-2 variants. A list of SARS-CoV-2 variant groups ({\bf
VG}) defined by WHO was initially considered. These variants were
assigned to a list of manually curated genomes from pango lineage
designation. We performed alignment and variant calling in all groups,
generating a total of 371 polymorphisms (343 SNPs, 28 InDels). Lastly,
polymorphisms with high frequency ($>$80\%) in each group were used in
an unsupervised clustering procedure. As a result, Figure
\ref{fig:markers} shows that this procedure was capable to define
clusters of polymorphic sites that are predominantly associated to
each SARS-CoV-2 {\bf VG}. We also compared the markers found by this
analysis with SNPs from the phylogenetic tree compiled by
\cite{pmid33972780} coupled with the respective frequency in each {\bf
VG}. The obtained SNPs allowed the identification of the same
SARS-CoV-2 {\bf VG}s (Supplementary Figure \ref{suppFig1}A) detected
by the former approach (Figure \ref{fig:markers}). Further, the
predictions made by using both markers are very similar (Supplementary
Figure \ref{suppFig1}B). We conclude that either the pango-designation
sequences or the phylogenetic tree \cite{pmid33972780} approach can
be used to select the polymorphic markers of SARS-CoV-2 variants
required by our method.

\begin{figure}[!htpb]
\centering
\includegraphics[width=10cm]{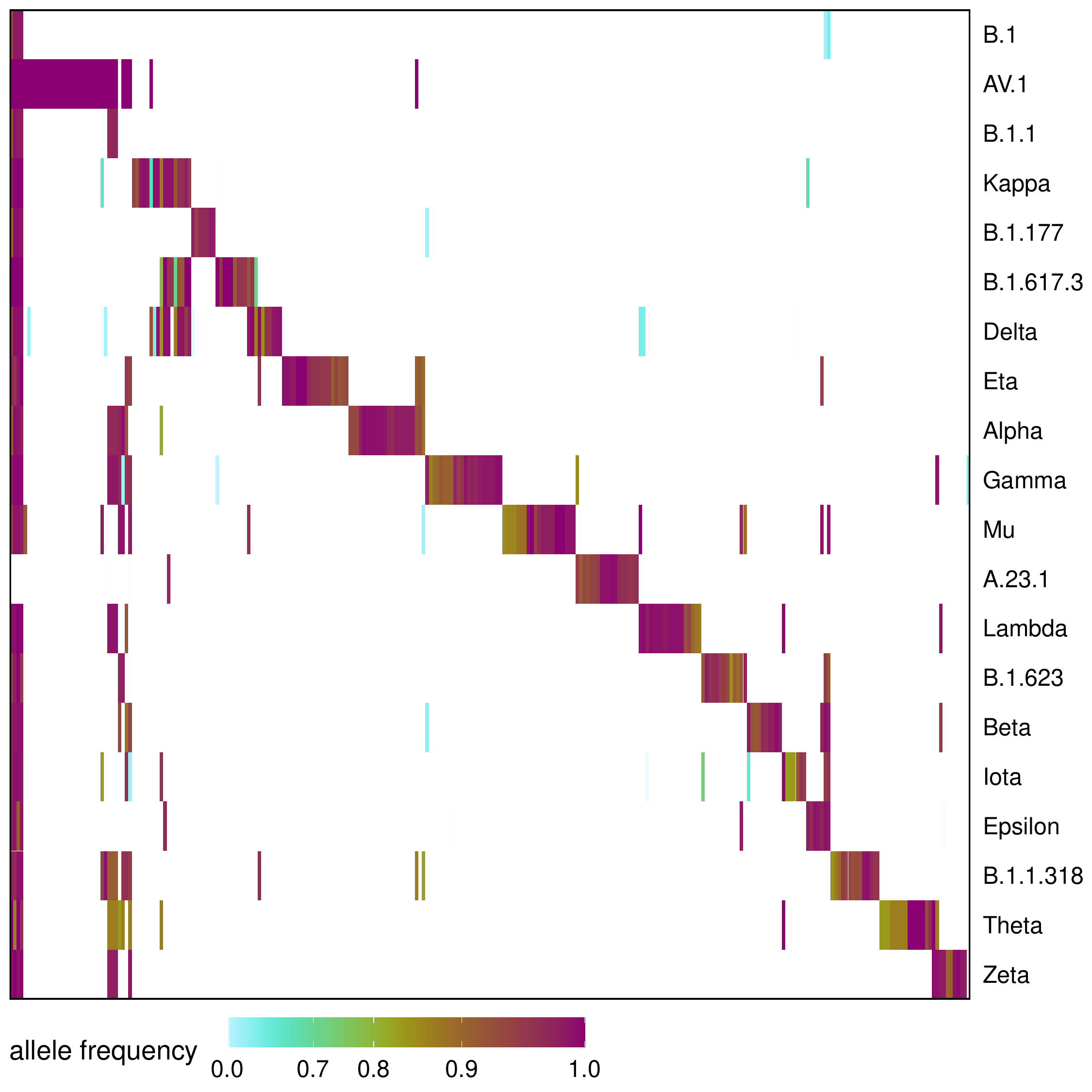}
\caption{Polymorphism-based markers across the SARS-CoV-2 {\bf
    VG}. The selected markers and SARS-CoV-2 {\bf VG} are listed
  respectively as columns and rows. Cells in  the heatmap are colored
  according to the relative frequency of each polymorphism.}
\label{fig:markers}
\end{figure}

\subsection{Simulation study}

To evaluate the performance of our method in predicting SARS-CoV-2
variant composition, we generated a synthetic dataset with 2000
simulated samples, considering non  uniform coverage. As shown by
Figure~\ref{fig:dryrun}, our method performed well when applied to
this data set. As expected,  estimations are more accurate when based
in variants found in higher relative frequencies (Figure
\ref{fig:dryrun}A). Mean absolute errors  (see \eqref{eqn:mae}) were
below 1\% in most cases, specially for variants with relative
frequencies above 25\% (Figure  \ref{fig:dryrun}B). Figure
\ref{fig:dryrun}C shows that the mean absolute error strongly depends
on sample coverage (the distribution of  simulated samples coverage
reflects the respective distribution on the wastewater dataset
considered in  Section~\ref{sec:Wastewater}). Finally, the adopted
bootstrap estimates of standard errors can provide accurate limits for
the true error, as shown in  Figure \ref{fig:dryrun}D. Results
obtained while analyzing few simulated data samples are summarized in
Table~\ref{tab:table-bootstrap}. In  particular, this table presents
how results are sensitive to sample coverage and composition
complexity, as well as the accuracy of  adopted bootstrap approach to
error estimation.

\begin{figure}[!htbp]
\centering
\includegraphics[width=12cm]{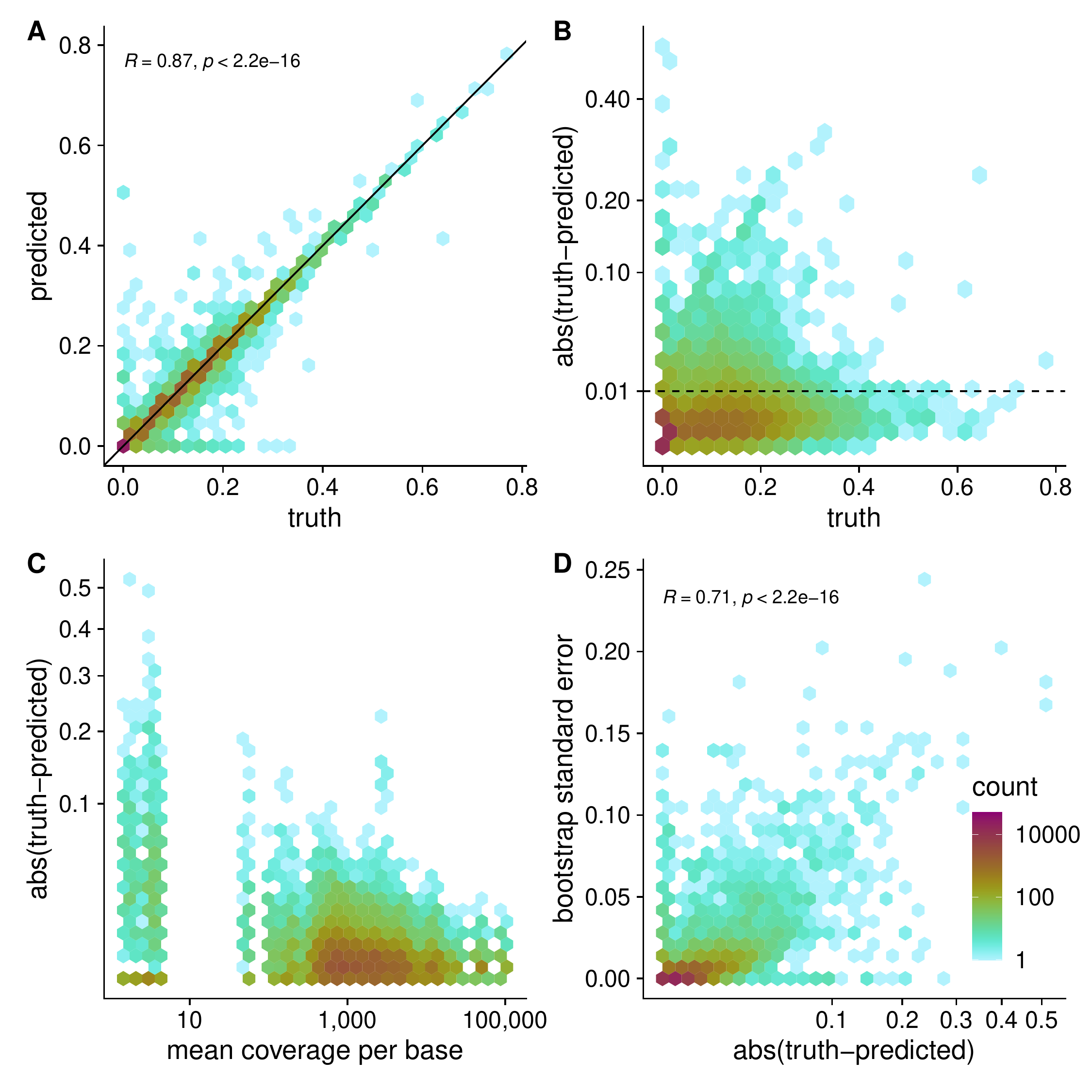}
\caption{Performance measures for the proposed estimator. The number
  of points within each hexagon is reflected by its color, according
  to the scale on the right side.  \textbf{A}. Model estimated
  proportions vs proportions actually used in
  simulations. \textbf{B}. Absolute errors in model estimations  vs
  simulated proportions. \textbf{C}. Mean absolute error in model
  estimations vs mean sample  coverage.  \textbf{D}. Bootstrap
  standard error for model estimations vs absolute error.}
\label{fig:dryrun}
\end{figure}

\begin{table}[]
\centering
\small
\caption{Results for 4 simulated samples. Mean Coverage (MC): average
  number of reads aligned to each considered polymorphism, that is
  $s^{-1} \sum_{i=1}^s t_i$. $w$: relative frequencies of each VG
  present in the simulation (ground truth); $\widehat{w}$: estimates
  of relative frequencies;  $\widehat{w}_{\text{boot}}$: estimated
  mean of bootstrap samples; $\widehat{\text{se}}_{\text{boot}}$:
  bootstrap standard error estimate for $\widehat{w}$.}
\label{tab:table-bootstrap} 
\begin{tabular}{cclllll}
\hline
Sample & MC & VG & $w$ & $\widehat{w}$ & $\widehat{w}_{\text{boot}}$ & $\widehat{\text{se}}_{\text{boot}}$ \\ \hline
\multirow{10}{*}{S\_53} & \multirow{10}{*}{2.03}    & B.1    & 0.123    & 0.187    & 0.112    & 0.156    \\
                                            &   & B.1.1.318             & 0.137  & 1.4e-08  & 1.63e-3  & 0.0163   \\
                                            &   & B.1.177               & 0.0411 & 0.0756   & 0.0708   & 0.0504   \\
                                            &   & B.1.617.3             & 0.0548 & 0.0627   & 0.111    & 0.138    \\
                                            &   & B.1.621               & 0.0822 & 0.0578   & 0.0503   & 0.0385   \\
                                            &   & Delta      & 0.137  & 0.215    & 0.196    & 0.0903   \\
                                            &   & Epsilon & 0.0822 & 0.181    & 0.176    & 0.0951   \\
                                            &   & Iota         & 0.0959 & 0.0316   & 0.0314   & 0.0319   \\
                                            &   & Lambda          & 0.137  & 0.122    & 0.12     & 0.0795   \\
                                            &   & Zeta             & 0.11   & 0.0682   & 0.0704   & 0.043    \\ \hline
\multirow{8}{*}{S\_1583}  & \multirow{8}{*}{2.94} & B.1                   & 0      & 0.498    & 0.377    & 0.181    \\
                                            &   & B.1.1                 & 0.323  & 6.16e-08 & 0.0275   & 0.106    \\
                                            &   & B.1.1.318             & 0.0323 & 4.7e-09  & 5.89e-09 & 3.52e-09 \\
                                            &   & B.1.177               & 0.0323 & 3.86e-08 & 0.014    & 0.0866   \\
                                            &   & B.1.617.3             & 0.258  & 0.0877   & 0.0959   & 0.048    \\
                                            &   & B.1.623               & 0.258  & 0.237    & 0.253    & 0.053    \\
                                            &   & Epsilon & 0      & 0.0849   & 0.125    & 0.125    \\
                                            &   & Gamma            & 0.0968 & 0.0916   & 0.0942   & 0.0325   \\ \hline
\multirow{17}{*}{S\_221} & \multirow{17}{*}{732.37} & Alpha & 0   & 6.43e-4      & 6.71e-4                    & 2.52e-4                            \\
                                            &   & AV.1                  & 0.0722 & 0.0748   & 0.0747   & 2.38e-3  \\
                                            &   & B.1                   & 0.0928 & 0.107    & 0.108    & 0.0206   \\
                                            &   & B.1.1                 & 0.0825 & 0.0588   & 0.0592   & 0.0203   \\
                                            &   & B.1.1.318             & 0.0515 & 0.048    & 0.048    & 2.19e-3  \\
                                            &   & B.1.177               & 0      & 1.35e-3  & 1.33e-3  & 7.32e-4 \\
                                            &   & B.1.617.3             & 0.103  & 0.105    & 0.106    & 3.39e-3  \\
                                            &   & B.1.621               & 0.103  & 0.104    & 0.104    & 2.37e-3  \\
                                            &   & B.1.623               & 0.103  & 0.104    & 0.104    & 2.59e-3  \\
                                            &   & Beta         & 0      & 1.64e-3  & 1.71e-3  & 5.7e-4  \\
                                            &   & Delta      & 0.103  & 0.101    & 0.101    & 3.34e-3  \\
                                            &   & Gamma            & 0.0722 & 0.0752   & 0.0752   & 2e-3    \\
                                            &   & Iota         & 0      & 4.11e-3  & 3.95e-3  & 1.37e-3  \\
                                            &   & Kappa      & 0.103  & 0.104    & 0.103    & 4.29e-3  \\
                                            &   & Lambda          & 0.0515 & 0.0522   & 0.0522   & 1.89e-3  \\
                                            &   & Theta            & 0      & 5.95e-4 & 5.87e-4 & 2.09e-4 \\
                                            &   & Zeta             & 0.0619 & 0.0575   & 0.0572   & 3.3e-3   \\ \hline
\multirow{6}{*}{S\_1645} & \multirow{6}{*}{24740.78} & AV.1                  & 0.353  & 0.353    & 0.353    & 5.63e-4 \\
                                            &   & B.1.621               & 0.118  & 0.117    & 0.117    & 3.91e-4 \\
                                            &   & Beta         & 0      & 6.98e-4 & 6.54e-4 & 2.22e-4 \\
                                            &   & Gamma            & 0.118  & 0.118    & 0.118    & 4.31e-4 \\
                                            &   & Iota         & 0.118  & 0.116    & 0.116    & 7.6e-4  \\
                                            &   & Theta            & 0.294  & 0.294    & 0.294    & 4.67e-4 \\ \hline

\end{tabular}%
\end{table}

\subsection{Wastewater data tests}\label{sec:Wastewater}

As SARS-CoV-2 variants are continuously spreading and evolving, the
environment surveillance has come into play in help bringing the
pandemic under control. Thus, considering that wastewater samples
provide a screenshot of circulating viral lineages in the community
\cite{pmid33024908}, we assess the reliability and utility of our
method to unveil the SARS-CoV-2 diversity from genomic sequencing data
of samples collected over time in two Swiss wastewater treatment
plants of Z\"urich (located in the canton of Z\"urich) and Lausanne
(located in the canton of Vaud),
\cite{Jahn2021}. Figure~\ref{suppFig2} in the Supplementary material
provides an overview of the polymorphism frequencies of all markers in
the pooled samples. We first recovered the relative frequencies of all
SARS-CoV-2 variants (Figure \ref{fig:density}A) from the cantons of
Z\"urich and Vaud, previously deposited in GISAID database. Next, we
compared these with the proportion of SARS-CoV-2 \emph{VG}s decomposed
by our method considering the viral sequencing reads obtained from the
wastewater samples (Figure \ref{fig:density}B). A \texttt{loess}
regression, implemented in the R statistical programming language, was
used to interpolate the proportion of variants between the missing
time periods in the actual longitudinal data samples (Figure
\ref{fig:density}C). The evolution of the inferred relative
frequencies for the Alpha variant are shown in Supplementary Figure
\ref{suppFig3}.  Results show the quick spread of the Alpha variant in
both cantons in early December 2020 and January 2021. The trend
observed in wastewater samples of both regions matches quite well the
one observed in the GISAID data for COVID-19 patients.

We also explored RNA samples of SARS-CoV-2 used to assess the
reproducibility of B.1.1.7 prevalence in a dilution series experiment
described in \cite{Jahn2021}. These samples contain a mixture of wild
type and Alpha/B.1.1.7 SARS-CoV-2 at ratios 10:1, 50:1 and 100:1, and
each one was sequenced five times. We observed that the estimated
composition is consistent with the respective dilutions (Supplementary
Figure \ref{suppFig4} and \ref{suppFig5}A). By merging the 5
replicates into a single sample, the overall coverage improves the
estimates of variant composition predictions (Supplementary Figure
\ref{suppFig5}B). However, we noted a small proportion of the {\bf VG}
\emph{A.23.1} and looking back at the marker heatmap (Supplementary
Figure \ref{suppFig4}), all dilution samples reveal a mutation
(\emph{S:V367F}) in high frequency which is a known marker of
\emph{A.23.1} \cite{pmid34163035}. Since the frequency of this
mutation is not consistent with the dilution amounts as expected, we
believe this is likely to be a sequencing artifact or contamination.

\begin{figure}[h]
\centering
\includegraphics[width=12cm]{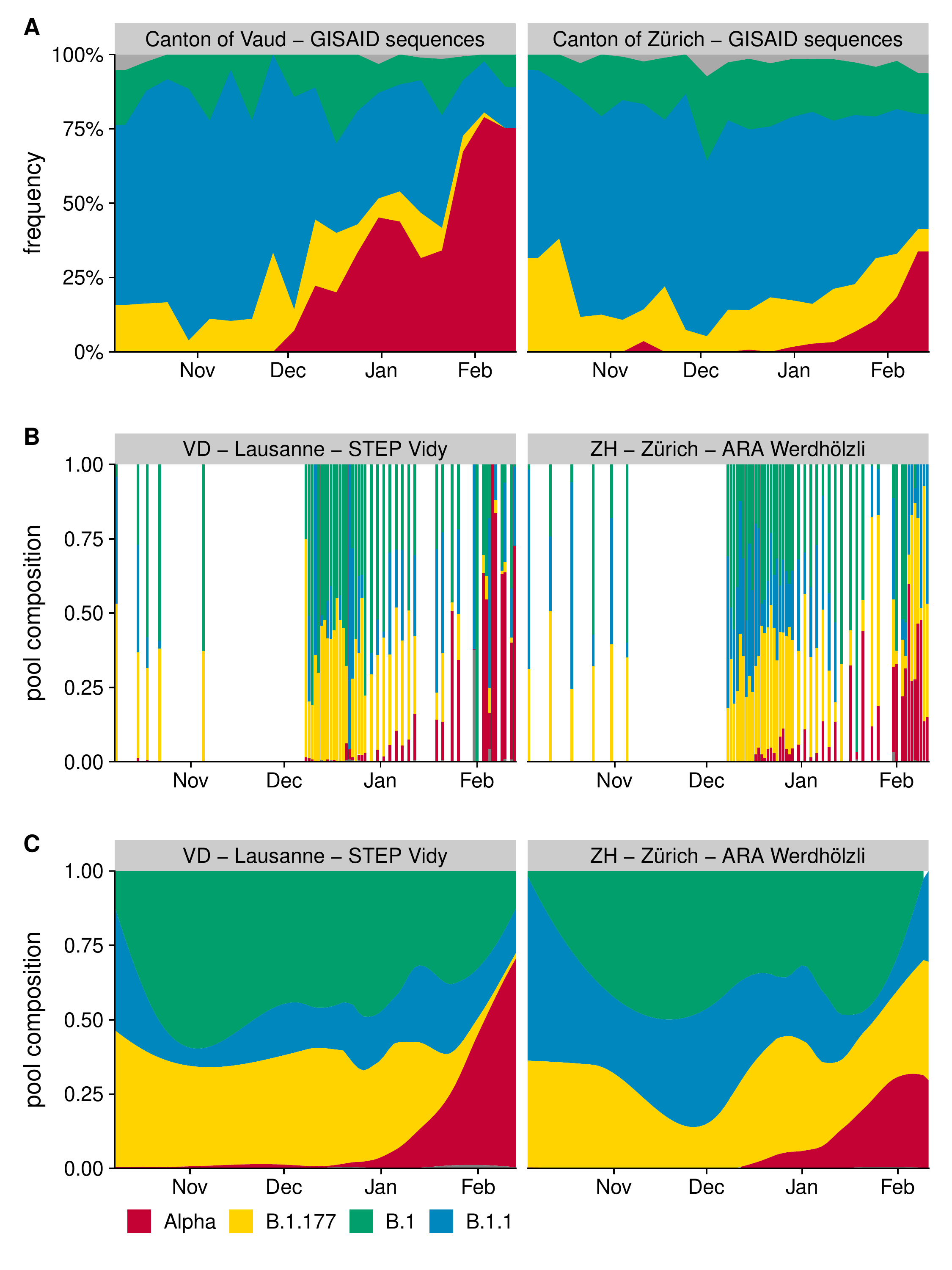}
\caption{(A) Area plot showing the longitudinal surveillance of Alpha
  (red) and other SARS-CoV-2 variants based on submitted sequences
  from {\bf GISAID} for the Cantons  of Vaud (left) and Z\"urich
  (right) between October 2020 and February 2021. (B) Viral
  composition estimated from sequencing data of  wastewater
  longitudinal samples collected in the sewage treatment plants of
  Lausanne (Canton of Vaud) and Z\"urich  (Canton of Zurich). (C)
  Evolution of viral composition obtained by a \texttt{loess}
  interpolation of (B). 
}
\label{fig:density}
\end{figure}

\section{Discussion}
The ability to implement continuous molecular surveillance of
SARS-CoV-2 has helped to accurately detect the prevalence of viral
strains. For example, measurement of SARS-CoV-2 RNA in wastewater has
been shown to be a useful tool to track SARS-CoV-2 and thus may help
to support public health policies. In addition to this approach, a
promising genomic surveillance initiative in the city-scale monitoring
encompasses the swabbing of surfaces in highly accessed locations
including hospitals, airports, parks, subway and bus stations. These
efforts will allow the tracking of diverse pathogens and, for COVID-19
they may accelerate the discovery of new variants and anticipate the
detection of variants of clinical interest, that may lead to new waves
of disease and the potential failure of some vaccines to offer
long-term protection. These large-scale efforts underscore the
development of novel analytical tools to identify the prevalence of
viral diversity from Next-Generation Sequencing data. 

This paper describes a new end-to-end approach to assess the presence
of SARS-CoV-2 variant groups throughout mixed DNA samples. First, we
propose a model considering the relative frequencies of polymorphic
markers found in samples positive for SARS-CoV-2 genomes, likely to be
derived from multiple subjects, allowing the determination of variant
frequencies. Then, we systematically evaluate its performance by
simulating different sequencing depths and variant relative
frequencies. These dry runs highlighted the method accuracy and the
sensibility of its performance to low coverage. 

Next we evaluate the estimation of the relative frequency of
SARS-CoV-2 lineages in public genomic sequencing data constituted by
122 wastewater samples from two cantons in Switzerland. Our results
show trends in conformity with data from SARS-CoV-2-positive clinical
samples, and recovered the evolution of the lineages observed in these
cities. Our findings endorse the utility of viral RNA monitoring in
municipal wastewater for SARS-CoV-2 infection surveillance at a
population-wide level. 

In addition, our pipeline uses multi-threading for efficient
parallelization, and is designed on a scalable workflow engine
\cite{pmid29788404}. The software provides a wrapper for marker
calling in a bioconda environment \cite{pmid29967506}, but it also
allows the user to load their own marker selection in the variant call
format ({\bf VCF}). The output consists of two files, including a flat
file with diversity estimation and a VCF file containing annotation of
other, non-marker polymorphisms, for further analysis. The flexibility
allowed by the choice of custom polymorphism-based markers,
considerably widens the scope of the tools described throughout,
allowing either the analysis of other viruses or regional
epidemiological studies. 

There are however shortcomings in our approach. The sequence coverage
across the viral genome is crucial for the detection of
polymorphism-based markers, and the precise determination of
SARS-CoV-2 variants. Given that the viral composition assessment
relies on polymorphic markers, we advise the use of a high sensitivity
variant caller to detect all relevant polymorphisms. Finally, the
estimation of SARS-CoV-2 variant composition may carry some
uncertainty due to the stochastic nature of pool sequencing. We
overcome this by considering a bootstrap approach to estimate standard
errors in predictions, thus providing a measure of the reliability of
each result. 

In summary, we present an useful method to decompose reliable
SARS-CoV-2 lineages using sequencing reads obtained from mixed
samples. The effectiveness of our method on both synthetic and real
data sets further demonstrates its utility for tracking SARS-CoV-2. We
believe this will translate into applied tools that will aid in the
environmental genomic surveillance efforts against COVID-19 outbreaks
or future pandemics.

\section*{Acknowledgments}
We gratefully acknowledge the authors of the data shared via
GISAID. ED-N is a researcher from Conselho Nacional de Desenvolvimento
Cient\'ifico e Tecnol\'ogico (CNPq, Brazil).

\section*{Funding}
This project was funded by Fundação de Amparo à Pesquisa do Estado de
São Paulo (FAPESP), grant 2021/05316-2.

\section*{Appendix}
\begin{proof}[Proof of Lemma~\ref{lem:loglik}] Substitution of the
expression for the law of $Z_i$ into $L(w)$ and observing that $t_i =
c_i^r+c_i^a$ gives
\begin{align*}
  L(w)
  &=
    \prod_i \prod_{j} \sum_{z_i\in \txcal{Z}}\frac{t_i!}{z_{ij}^a!
    z_{ij}^r!} (P_{ij}w_j)^{z_{ij}^a}  ((1-  P_{ij})w_j)^{z_{ij}^r} \\
  &=
    \prod_i \sum_{z_i \in \txcal{Z}} \frac{t_i!}{c_i^a! c_i^r!}
    \frac{c_i^a! \prod_j  (P_{ij}w_j)^{z_{ij}^a}}{\prod_j z_{ij}^a!}
    \frac{c_i^r!  \prod_j ((1-
    P_{ij})w_j)^{z_{ij}^r}}{\prod_j z_{ij}^r!}  \\
  &=
  \prod_i \frac{t_i!}{c_i^a ! c_i^r !}
  \bigg(\sum_j P_{ij}w_j\bigg)^{c_i^a }
  \bigg(1-\sum_j P_{ij}w_j\bigg)^{c_i^r } \\
  &\qquad\qquad \times
  \sum_{z_i \in \txcal{Z}} \bigg[ \frac{c_i^a !}{\prod_j
    z_{ij}^a!}  \bigg(\frac{P_{ij}w_j}{\sum_j P_{ij}w_j}\bigg)^{c_i^a}  
    \frac{c_i^r !}{\prod_j z_{ij}^r!}
  \bigg(\frac{(1- P_{ij})w_j}{1-\sum_j P_{ij}w_j}\bigg)^{c_i^r } \bigg] \\
  &=
  \prod_i \frac{t_i!}{c_i^a ! c_i^r !}
  \bigg(\sum_j P_{ij}w_j\bigg)^{c_i^a }
  \bigg(1-\sum_j P_{ij}w_j\bigg)^{c_i^r } 
  \sum_{z_i^a \in \txcal{Z}} \frac{c_i^a !}{\prod_j
    z_{ij}^a!}\bigg(\frac{P_{ij}w_j}{\sum_j P_{ij}w_j}\bigg)^{c_i^a}   \\
  &\qquad\qquad\times
  \sum_{z_i^r\in\txcal{Z}} \frac{c_i^r !}{\prod_j
    z_{ij}^r!}\bigg(\frac{(1- P_{ij})w_j}{1-\sum_j
    P_{ij}w_j}\bigg)^{c_i^r}  \\
  &=
  \prod_i \frac{t_i!}{c_i^a ! c_i^r !}
  \bigg(\sum_j P_{ij}w_j\bigg)^{c_i^a}
  \bigg(1-\sum_j P_{ij}w_j\bigg)^{c_i^r}
\end{align*}
The assertion made by the lemma follows by considering the logarithm
of the last expression. 
\end{proof}

\begin{proof}[Proof of Lemma~\ref{lem:convexity}]
Throughout, let $e$ and $o$ be any two indices in $\{1, 2, \ldots,
v\}$.  For any sample $c$ and $i \in \{1$, $2$, $\ldots$, $s\}$ define
\[
  Q_i(w)
  =
  \frac{c_i^a}{(\sum_{j=1}^v P_{ij}w_j)^2} +
  \frac{c_i^r}{(1-\sum_{j=1}^v P_{ij}w_j)^2}.
\]
Straightforward computations yield
\[
  \frac{\partial^2}{\partial w_e^2} \ell(w)
  =
 -\sum_{i=1}^sP_{ie}^2 Q_i(w)
\]
and
\[
\frac{\partial^2}{\partial w_e\partial w_o} \ell(w)
  =
 -\sum_{i=1}^s P_{ie}P_{io} Q_i(w).
\]
Let $H$ be the Hessian of $\ell(w)$ with respect to $w$, namely the
$v \times v$ matrix with entries $H_{eo} = \partial^2
\ell(w)/(\partial w_e \partial w_o)$. For any $u \in \sS$ it follows
that
\begin{align*}
  u^tHu
  &=  \sum_{e,o} u_{e}H_{eo}u_{o}
   = - \sum_{e,o}
    u_{e}u_{o}\sum_{i=1}^sP_{ie}P_{io}Q_i(w)\\
  &= -\sum_{i=1}^sQ_i(w) \sum_{e,o}
    u_{e}P_{ie}P_{io}u_{o}
   = -\sum_{i=1}^sQ_i(w) \bigg(\sum_e P_{ie}u_e\bigg)^2.
\end{align*}
The fact that $Q_i(w)\geq 0$ for all $i$ and $w$ leads to $H$
being negative semi-definite, thus concluding the proof.
\end{proof}

\bibliographystyle{alpha}

\newpage

\section*{Supplementary matterial}

\begin{figure}[h]
\centering
    \includegraphics[width=.9\linewidth]{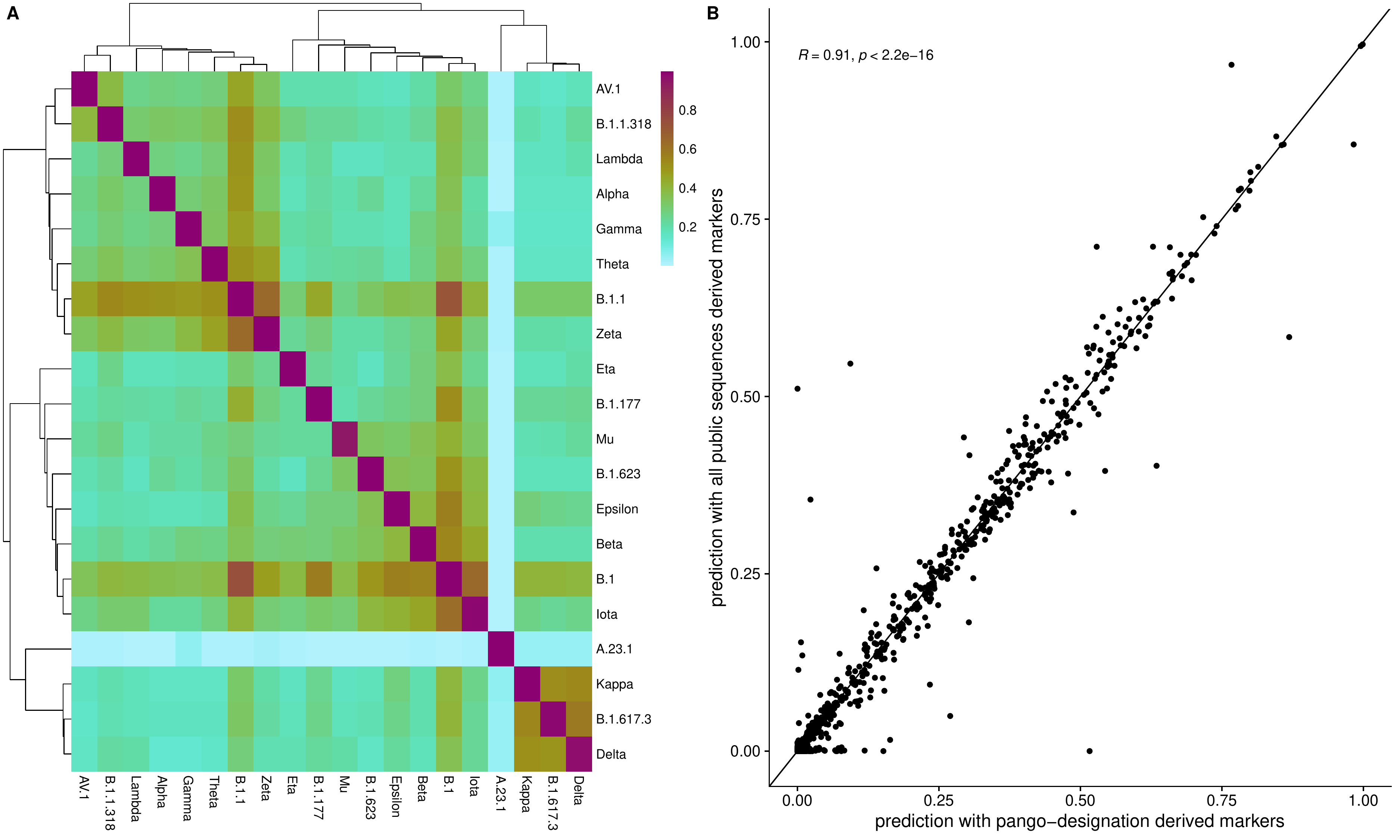}
    \caption[S1]{(A) Heatmap showing the similarities of polymorphism
      frequencies in SARS-CoV-2 variant groups, comparing
      pango-designation based markers (horizontal axis) and the
      equivalent markers from phylogenetic tree \cite{pmid33972780}
      (vertical axis). Similarities were calculated over polymorphisms
      present on both sets and measured by cosine similarity, which
      ranges between 0 and 1 (1 meaning a perfect match). (B) Scatter
      plot displaying model predictions for variant frequencies on
      wastewater samples, either obtained with pango-designation based
      markers (horizontal axis) or markers derived from all public
      sequences on {\bf UCSC} phylogenetic tree
      \footnote{\url{https://hgdownload.soe.ucsc.edu/goldenPath/wuhCor1/UShER_SARS-CoV-2//}}
      (vertical axis).} 
\label{suppFig1}
\end{figure}

\pagebreak

\begin{figure}
\centering
\includegraphics[width=0.95\textwidth]{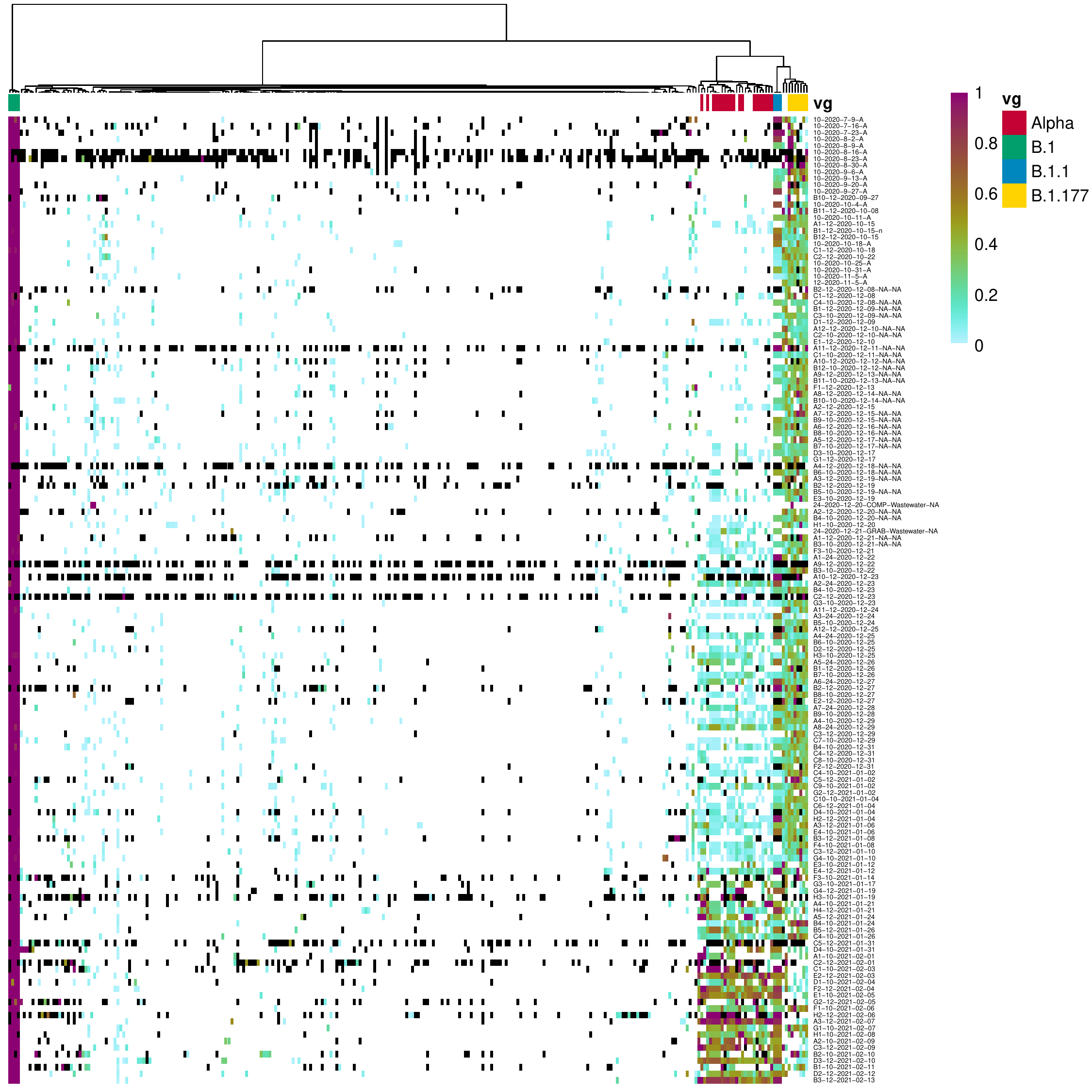}
\caption[S2]{Polymorphisms frequencies in all wastewater samples from
  the Swiss dataset, ordered row-wise by collection date, with the
  respective markers from the 4 major {\bf VG} annotated. Black cells
  denote no coverage in the respective sample. } 
\label{suppFig2}
\end{figure}

\pagebreak

\begin{figure}
\centering
 \includegraphics[width=0.85\textwidth]{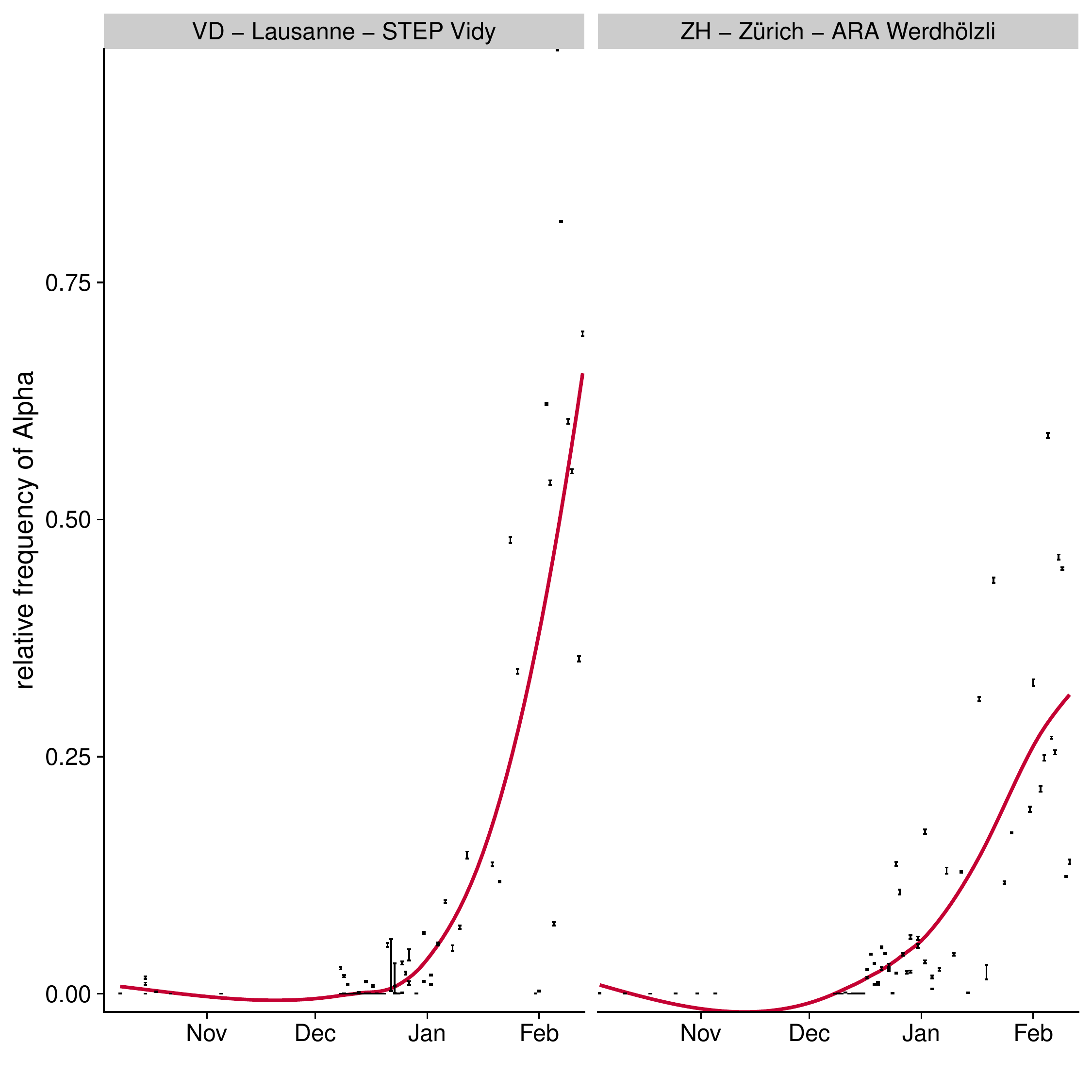}
\caption{Graphic based on wastewater sequencing data collected over
  time (122 samples) in the sewage treatment plants of Lausanne and
  Zürich. \texttt{LOESS} (locally estimated scatter plot smoothing) of
  the Alpha {\bf VG} abundance, with error bars denoting the bootstrap
  standard errors.} 
\label{suppFig3}
\end{figure}

\begin{figure}
\centering
\includegraphics[width=0.9\textwidth]{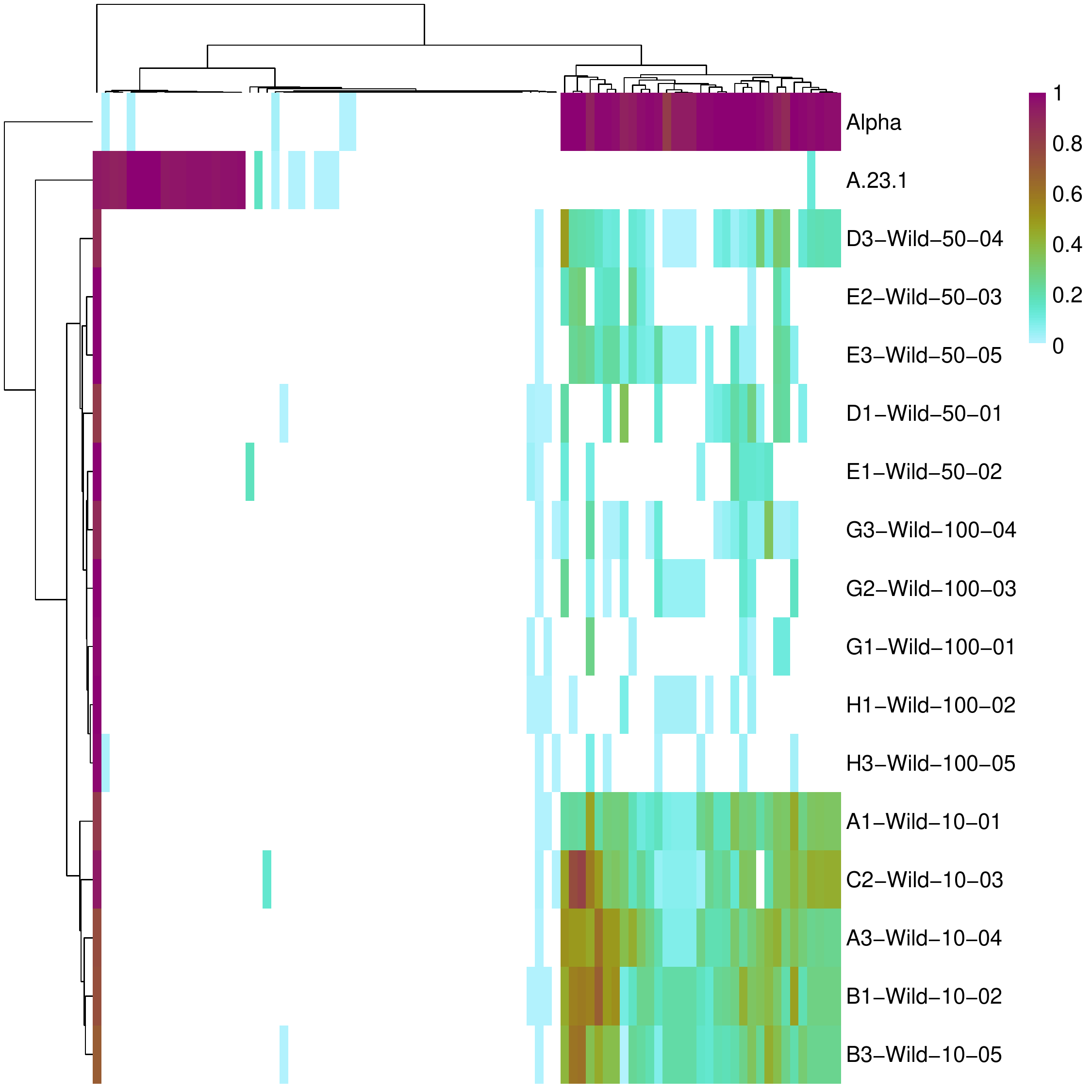}
\caption{Heatmap result of an unsupervised hierarchical clustering of
  polymorphism frequencies in the dilution samples. Alpha and A.23.1
  {\bf VG} markers were added for comparison.} 
\label{suppFig4}
\end{figure}

\begin{figure}
\centering
\includegraphics[width=0.9\textwidth]{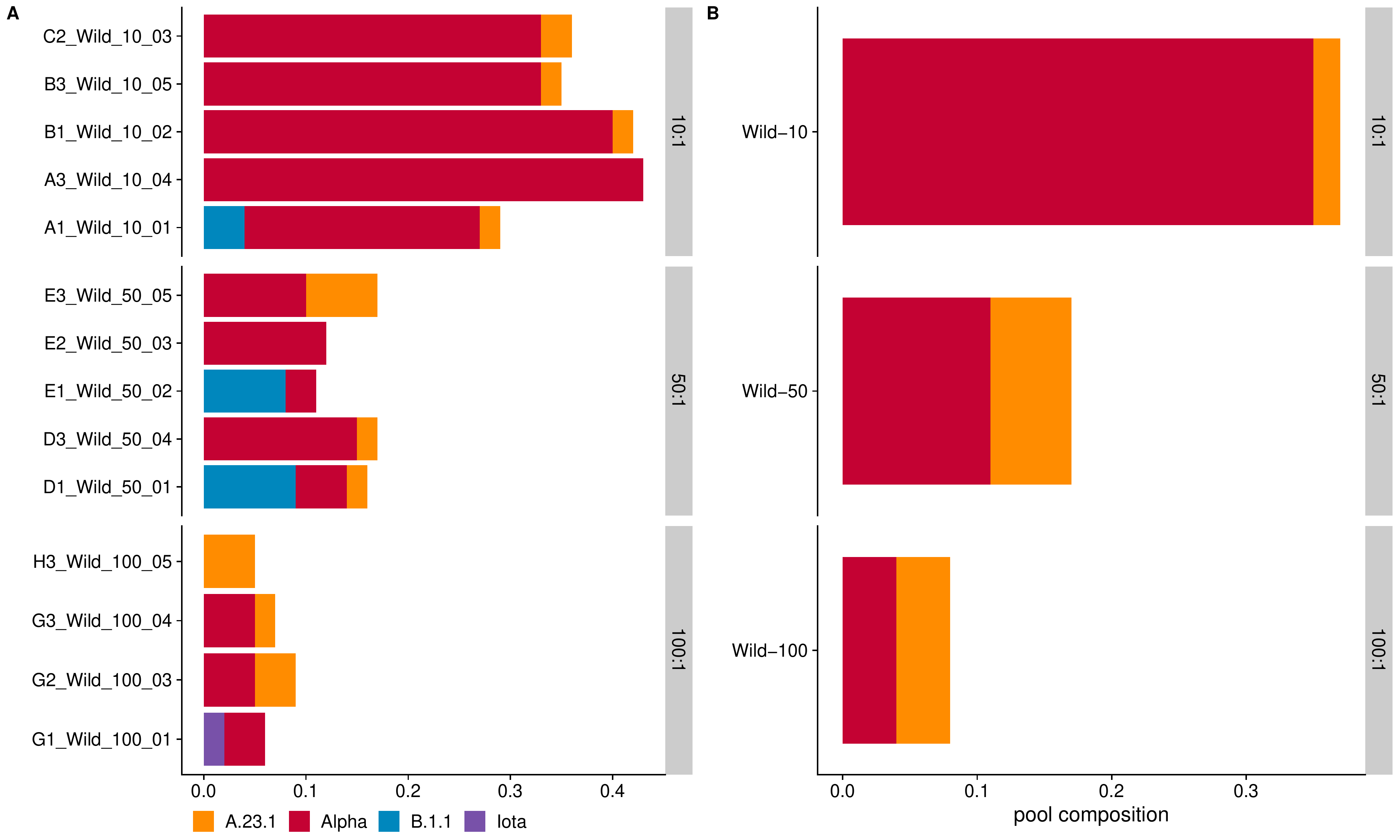}
\caption{(A) Model predictions for each of the dilution samples
  replicates. (B) Model predictions after merging the replicates into
  a single sample.} 
\label{suppFig5}
\end{figure}

\begin{table}[]
\centering
\scriptsize
\caption{Variant Groups}
\label{tab:table-vg}
\begin{tabular}{ll}
\hline
\rowcolor[HTML]{EFEFEF}
Variant Group & Pango Lineages \\ \hline
A.23.1 & A.23.1 \\
AV.1 & AV.1 \\
B.1.1.318 & B.1.1.318 \\
Alpha & B.1.1.7 \\
Beta & B.1.351, B.1.351.2, B.1.351.3 \\
Epsilon & B.1.427, B.1.429 \\
Eta & B.1.525 \\
Iota & B.1.526 \\
Kappa & B.1.617.1 \\
Delta & B.1.617.2, AY.1, AY.2, AY.3, AY.3.1 \\
B.1.617.3 & B.1.617.3 \\
Mu & B.1.621, B.1.621.1 \\
B.1.623 & B.1.623 \\
Lambda & C.37 \\
Gamma & P.1, P.1.1, P.1.2, P.1.4, P.1.6, P.1.7 \\
Zeta & P.2 \\
Theta & P.3 \\
B.1.177 & B.1.177 \\
B.1.1 & B.1.1 \\
B.1 & B.1 \\
\hline
\end{tabular}
\end{table}

\end{document}